\documentclass[conference]{IEEEtran}

\usepackage{cite}      

\usepackage[usenames,dvipsnames,svgnames,table]{xcolor}
\usepackage{graphicx}  

\usepackage{algorithm}
\usepackage{algpseudocode}

\usepackage{psfrag}    

\usepackage{url}       
\usepackage{tikz}
\usepackage{pgfplots}
\pgfplotsset{width=7cm,compat=1.3}

\usepackage{psfrag}
\usepackage{makecell}

\usepackage{amssymb}
\usepackage{mathtools}

\usepackage{amsmath}   
\usepackage{amssymb}

\usepackage{xfrac}

\newcommand{\nop}[1]{} 
\newcommand{\shorten}[1]{}

\newtheorem{theorem}{Theorem}
\newtheorem{definition}{Definition}

\newtheorem{lemma}{Lemma}

\newtheorem{corollary}{Corollary}
\newtheorem{example}{Example}

\newtheorem{construction}{Construction}
\newcommand{\signed}%
    {{\unskip\nobreak\hfill\penalty50
      \hskip2em\hbox{}\nobreak\hfil $\blacksquare$
      \parfillskip=0pt \finalhyphendemerits=0 \par}}
\newenvironment{proof}[1]
    {
    \bf{Proof:}\rm{\noindent{#1 }}\ignorespaces
    }
    {\signed\addvspace\medskipamount}

\begin{document}

\title{Repair Duality with Locally Repairable and Locally Regenerating Codes}
%
\author{
	\IEEEauthorblockN{Danilo Gligoroski\IEEEauthorrefmark{1}, Katina Kralevska\IEEEauthorrefmark{1}, Rune E.~Jensen\IEEEauthorrefmark{2}, and Per Simonsen\IEEEauthorrefmark{3}\\ Email: danilog@ntnu.no, katinak@ntnu.no, runeerle@idi.ntnu.no, per.simonsen@memoscale.com}
	\IEEEauthorblockA{\IEEEauthorrefmark{1}Department of Information Security and Communication Technology, NTNU, Norwegian University of Science and Technology}
	\IEEEauthorblockA{\IEEEauthorrefmark{2}Department of Computer Science, NTNU, Norwegian University of Science and Technology}
	\IEEEauthorblockA{\IEEEauthorrefmark{3}MemoScale AS, Norway}
}

\maketitle

\begin{abstract}	
	We construct an explicit family of locally repairable and locally regenerating codes whose existence was proven in a recent work by Kamath et al. about codes with local regeneration but no explicit construction was given. This explicit family of codes is based on HashTag codes. HashTag codes are recently defined vector codes with different vector length $\alpha$ (also called a sub-packetization level) that achieve the optimal repair bandwidth of MSR codes or near-optimal repair bandwidth depending on the sub-packetization level. We applied the technique of parity-splitting code construction. We show that the lower bound on the size of the finite field for the presented explicit code constructions can be lower than the one given in the work of Kamath et al. Finally, we discuss the importance of having two ways for node repair with locally regenerating HashTag codes: repair only with local parity nodes or repair with both local and global parity nodes. To the best of the authors' knowledge, this is the first work where this duality in repair process is discussed. We give a practical example and experimental results in Hadoop where we show the benefits of having this repair duality.
\end{abstract}

{\bfseries {Keywords}}: Vector codes, Repair bandwidth, Repair locality, Exact repair, Parity-splitting, Hadoop.

%
\IEEEpeerreviewmaketitle

\section{Introduction}
The repair efficiency of erasure coding in distributed storage systems is measured with two main metrics: the amount of transferred data during a repair process (\emph{repair bandwidth}) and the number of accessed nodes in a repair process (\emph{repair locality}). Two types of erasure codes that are optimal with respect to these two metrics have emerged: \emph{Regenerating Codes} (RCs) \cite{5550492} and \emph{Locally Repairable Codes} (LRCs) \cite{journals/tit/GopalanHSY12,conf/infocom/OggierD11,6195703}.

Minimum Storage Regenerating (MSR) codes are an important class of RCs that minimize the amount of data stored per node and the repair bandwidth. Explicit constructions of MSR codes can be found in \cite{7463553,journals/corr/KralevskaGJO16}.
LRCs relax the maximum distance separable (MDS) requirement in order to minimize the number of nodes accessed during a repair. Studies on implementation and performance evaluation of LRCs can be found in \cite{Huang:2013:PCF:2435204.2435207,conf/usenix/HuangSXOCG0Y12,journals/pvldb/SathiamoorthyAPDVCB13,7593121}. Combining the benefits of RCs and LRCs in one storage system can bring huge savings in practical implementations. 
For instance, repair bandwidth savings by RCs are important when repairing huge amounts of data, while a fast recovery and an access to small number of nodes enabled by LRCs are desirable for repair of frequently accessed data. 

Several works present code constructions that combine the benefits of RCs and LRCs \cite{6655894,6846301,7282575}. 
Rawat et al. in \cite{6655894} and Kamath et al. in \cite{6846301} have independently investigated codes with locality in the context of vector codes, and they call them locally repairable codes with local minimum storage regeneration (MSR-LRCs) and codes with local regeneration, respectively. 
Rawat et al. \cite{6655894} provided an explicit construction, based on Gabidulin maximum rank-distance codes, of vector linear codes with all-symbol locality for the case when the local codes are MSR codes. However, the complexity of these codes increases exponentially with the number of nodes due to the two-stage encoding. 
In \cite{6846301}, Kamath et al. gave an existential proof without presenting an explicit construction. 
Another direction of combining RCs and LRCs is to use repair locality for selecting the accessed nodes in a RC \cite{7282575}, while an interpretation of LRCs as exact RCs was presented in \cite{7541379}.
Two different erasure codes, product and LRC codes, are used to optimize for recovery performance and reduce the storage overhead in \cite{188446}.

\textbf{Our Contribution:}
We construct an explicit family of locally repairable and locally regenerating codes whose existence was proven in a recent work by Kamath et al. \cite{6846301} about codes with local regeneration. In that work, an existential proof was given, but no explicit construction was given. Our explicit family of codes is based on HashTag codes \cite{7463553,journals/corr/KralevskaGJO16}.
HashTag codes are MDS vector codes with different vector length $\alpha$ (also called a sub-packetization level) that achieve the optimal repair bandwidth of MSR codes or near-optimal repair bandwidth depending on the sub-packetization level.
We apply the technique of parity-splitting of HashTag codes in order to construct codes with locality in which the local codes are regenerating codes and which hence, enjoy both advantages of locally repairable codes as well as regenerating codes.

We also show (although just with a concrete example) that the lower bound on the size of the finite field where these codes are constructed, given in the work by Kamath et al. \cite{6846301}, can be lower. The presented explicit code construction has a practical significance in distributed storage systems as it provides system designers with greater flexibility in terms of selecting various system and code parameters due to the flexibility of HashTag code constructions.

Last but not least, we discuss the repair duality and its importance. Repair duality is a situation of having two ways to repair a node: to repair it only with local parity nodes or repair it with both local and global parity nodes.   To the best of the authors' knowledge, this is the first work that discusses the repair duality and how it can be applied based on concrete system and code parameters.

The paper is organized as follows. Section \ref{mathPrel} presents mathematical preliminaries. In Section \ref{paritySplitting}, we describe a framework for explicit constructions of locally repairable and locally regenerating codes. The repair process is analyzed in Section \ref{repair} where we explain the repair duality. Experimental results of measurements in Hadoop are given in Section \ref{Hadoop}. Conclusions are summarized in Section \ref{summary}.

\section{Mathematical Preliminaries and Related Work}\label{mathPrel}
Inspired by the work of Gopalan et al. about locally repairable codes \cite{journals/tit/GopalanHSY12}, Kamath et al. extended and generalized the concept of locality  in \cite{6846301}. In this paper, we use notation that is mostly influenced (and adapted) from those two papers.

\begin{definition}[\cite{6846301}]
	A $\mathbb{F}_q$-linear vector code of block length $n$ is a code $\mathcal{C} \in \big(\mathbb{F}_q^\alpha \big)^n$ having a symbol alphabet $\mathbb{F}_q^\alpha$ for some $\alpha \geq  1$, i.e. 
	$$\mathcal{C} = \{ \mathbf{c} = (\mathbf{c}_1, \mathbf{c}_2, \ldots, \mathbf{c}_n), \mathbf{c}_i \in \mathbb{F}_q^\alpha \text{ for all } i \in [n] \}, $$ 
	and satisfies the additional property that for given $\mathbf{c},\mathbf{c}' \in \mathcal{C}$ and $a,b \in \mathbb{F}_q$, 
	$$a \mathbf{c} + b\mathbf{c}' = (a \mathbf{c}_1 + b\mathbf{c}_1', a \mathbf{c}_2 + b\mathbf{c}_2', \ldots, a \mathbf{c}_n + b\mathbf{c}_n' )$$ also belongs to $\mathcal{C}$ where $a\mathbf{c}_i$ is a scalar multiplication of the vector $\mathbf{c}_i$. $\blacksquare$
\end{definition}

Throughout the paper, we refer to the vectors $\mathbf{c}_i$ as vector symbols or nodes. Working with systematic codes, it holds that for the systematic nodes $\mathbf{c}_i = \mathbf{d}_i$ for $1\leq i\leq k$ and for the parity nodes $\mathbf{c}_{k+i} = \mathbf{p}_i$ for $1\leq i\leq r$. For every vector code $\mathcal{C} \in \big(\mathbb{F}_q^\alpha \big)^n$ there is an associated scalar linear code $\mathcal{C}^{(s)}$ over $\mathbb{F}_q$ of length $N=\alpha n$. Accordingly, the dimension of the associated scalar code $\mathcal{C}^{(s)}$ is $K = \alpha k$. For a convenient notation, the generator matrix $G$ of size ${K \times N}$ of the scalar code $\mathcal{C}^{(s)}$ is such that each of the $\alpha$ consecutive columns corresponds to one code symbol $\mathbf{c}_i, i \in [n]$, and they are considered as $n$ thick columns $\mathbf{W}_i,  i \in [n]$. For a subset $\mathcal{I} \subset [n]$ we say that it is an information set for $\mathcal{C}$ if the restriction $G |_\mathcal{I}$ of $G$ to the set of thick columns with indexes lying in $\mathcal{I}$ has a full rank, i.e. $\text{rank}(G|_\mathcal{I}) = K$. 

The minimum cardinality of an information set is referred as quasi-dimension $\kappa$ of the vector code $\mathcal{C}$. As the vector code $\mathcal{C}$ is  $\mathbb{F}_q$-linear, the minimum distance $d_{\min}$ of $\mathcal{C}$ is equal to the minimum Hamming weight of a non-zero codeword in $\mathcal{C}$. Finally, a vector code of block length $n$, scalar dimension $K$, minimum distance $d_{\min}$, vector-length parameter $\alpha$ and quasi-dimension $\kappa$ is shortly denoted with $[n,K,d_{\min},\alpha,\kappa]$. While in the general definition of vector codes in \cite{6846301} the quasi-dimension $\kappa$ does not necessarily divide the dimension $K$ of the associated scalar, for much simpler and convenient description of the codes in this paper we take that $k=\kappa$, i.e. $K=\alpha \kappa$. In that case the erasure and the Singleton bounds are given by: 
\begin{equation}\label{SingletonVector}
d_{\min} \leq n - \kappa + 1.
\end{equation} 
In \cite{5550492}, Dimakis et al. studied the repair problem in a distributed storage system where a file of $M$ symbols from a finite field $\mathbf{F}_q$ is stored across $n$ nodes, each node stores $\frac{M}{k}$ symbols. They introduced the metric repair bandwidth $\gamma$, and proved that the repair bandwidth of a MDS code is lower bounded by
\begin{equation}
\gamma \geq \frac{M}{k}\frac{d}{d-k+1},
\label{optimalBW}
\end{equation}
where $d$ is the number of accessed available nodes (helpers).

\begin{lemma}[\cite{5550492}] The repair bandwidth of a $(n, k)$ MDS code is minimized for $d=n-1$. MSR codes achieve the lower bound of the repair bandwidth equal to
	\begin{equation}
	\gamma_{MSR}^{min}=\frac{M}{k}\frac{n-1}{n-k}.
	\label{optimalMSR}
	\end{equation}
\end{lemma}
\nop{
\begin{lemma}[\cite{5550492}] A $(n, k)$ MSR code attains the minimum storage point of the optimal tradeoff curve between the storage and the repair bandwidth, i.e.,
	\begin{equation}
	(\alpha_{MSR}, \gamma_{MSR}^{min})=(\frac{M}{k}, \frac{M}{k} \frac{n-1}{n-k}),
	\label{optimalMSR}
	\end{equation}
	when $n-1$ helper nodes are contacted.
\end{lemma}
}
A $(n, k)$ MSR code has the maximum possible distance $d_{min}=n-k+1$ in addition to minimizing the repair bandwidth, but it has the worst possible locality.
\begin{corollary} The locality of a $(n, k)$ MSR code is equal to $n-1$.
\end{corollary}

Any $[n,K,d_{\min},\alpha,\kappa]$ vector code $\mathcal{C}$ is MDS if and only if its generator matrix can be represented in the form $G = [I | P]$, where the $K \times (N - K)$ parity matrix 
		\begin{equation}\label{parityPart}
		\mathbf{P}=
		\begin{bmatrix}
		G_{1, 1} & G_{1, 2} & \ldots & G_{1, \kappa}\\
		G_{2, 1} & G_{2, 2} & \ldots & G_{2, \kappa}\\
		\vdotswithin{1} & \vdotswithin{\alpha_n} & \ddots & \vdotswithin{{\alpha_n}^{k-1}}\\
		G_{\kappa, 1} & G_{\kappa, 2} & \ldots & G_{\kappa, n-\kappa}\\
		\end{bmatrix},
		\end{equation}	
possesses the property that every square block submatrix of $P$ is invertible. The $G_{i,j}$ entries are square sub-matrices of size $\alpha \times \alpha$, and a block submatrix is composed by different entries of $G_{i,j}$.

In order to analyze codes with local regeneration, Kamath et al. introduced a new family of vector codes called uniform rank-accumulation (URA) codes in \cite{6846301}. They showed that exact-repair MSR codes belong to the class of URA codes.

\begin{definition}\cite[Def. 2]{6846301}
	\label{KamathInformationLocality}
	Let $\mathcal{C}$ be a $[n,K,d_{\min},\alpha,\kappa]$ vector code with a generator matrix $G$. The code $\mathcal{C}$ is said to have $(l,\delta)$ information locality if there exists a set of punctured codes $\{\mathcal{C}_i \}_{i \in \mathcal{L}}$ of $\mathcal{C}$ with respective supports $\{S_i \}_{i \in \mathcal{L}}$ such that
	\begin{itemize}
		\item $|S_i| \leq l+\delta - 1,$
		\item $d_{\min}(\mathcal{C}_i)\geq \delta,$ and
		\item $\text{rank}(G|_{\bigcup_{i \in \mathcal{L}}})=K.$
	\end{itemize}
\end{definition}

If we put $\delta=2$ in Def.\ref{KamathInformationLocality}, then we get the definition of information locality introduced by Gopalan et al. \cite{journals/tit/GopalanHSY12}. 
They derived the upper bound for the minimum distance of a $(n, k, d)_q$ code with information locality $l$ for $\delta=2$ as 
	\begin{equation}
	d_{\min} \leq n-k-\left\lceil \frac{k}{l} \right\rceil+2.
	\label{distance1}
	\end{equation}

A general upper bound was derived in \cite{6846301} as
	\begin{equation}
	d_{\min} \leq n-k+1-\left( \left\lceil \frac{k}{l} \right\rceil - 1 \right)(\delta - 1) .
	\label{distance}
	\end{equation}

Huang et al. showed the existence of Pyramid codes that achieve the minimum distance given in (\ref{distance1}) when the field size is big enough \cite{Huang:2013:PCF:2435204.2435207}.
Finally, based on the work by Gopalan et al. \cite{journals/tit/GopalanHSY12} and Pyramid codes by Huang et al. \cite{Huang:2013:PCF:2435204.2435207}, Kamath et al. proposed a construction of codes with local regeneration based on a parity-splitting strategy in \cite{6846301}. 
\shorten{
\begin{figure}
	\begin{minipage}[b]{0.5\linewidth}
		\centering
		\includegraphics[width=8.9cm,height=5cm]{9-6-alfa09}
	\end{minipage}
	\caption{A systematic generator matrix of the associated scalar code. Here the black squares on the main diagonal represent the value 1, but the black squares in the parity parts represent the non-zero values in $\mathbf{F}_{32}$. Note that if the parity matrix is partitioned in $9\times9$ square submatrices, it has the same form as in equation (\ref{parityPart}).}
	\label{G0906alpha09}
\end{figure}
}

\section{Codes with Local Regeneration from HashTag Codes by Parity-Splitting}\label{paritySplitting}

In \cite{7463553,journals/corr/KralevskaGJO16}, a new class of vector MDS codes called HashTag codes is defined. HashTag codes achieve the lower bound of the repair bandwidth given in (\ref{optimalMSR}) for $\alpha=r^{\lceil\frac{k}{r}\rceil}$, while they have near-optimal repair bandwidth for small sub-packetization levels. HashTag codes are of a great practical importance due to their properties: flexible sub-packetization level, small repair bandwidth, and optimized number of I/O operations.  
We briefly give the basic definition of HashTag codes before we construct codes with local regeneration from them by using the framework of parity-splitting discussed in \cite{6846301}.

\begin{definition}\label{HashTagCodes}
	A $(n,k,d)_q$ HashTag linear code is a vector systematic code defined over an alphabet $\mathbb{F}_q^\alpha$ for some $\alpha \geq  1$. It encodes a vector $\mathbf{x} = (\mathbf{x}_1,\ldots,\mathbf{x}_k)$, where $\mathbf{x}_i = (x_{1,i}, x_{2,i},\ldots,x_{\alpha,i})^T \in \mathbf{F}_q^\alpha$ for $i \in [k] $, to a codeword $\mathcal{C}(\mathbf{x}) = \mathbf{c} = (\mathbf{c}_1, \mathbf{c}_2, \ldots, \mathbf{c}_n)$ where the systematic parts $\mathbf{c}_i=\mathbf{x}_i$ for $i \in [k]$ and the parity parts $\mathbf{c}_i=(c_{1,i}, c_{2,i},\ldots,c_{\alpha,i})^T$ for $i \in \{k+1,\ldots,n\}$ are computed by the linear expressions that have a general form as follows:
	\begin{equation}\label{LinEquations}
	c_{j,i}=\sum f_{\nu,j, i} x_{j_1,j_2},
	\end{equation}
	where $f_{\nu,j, i}\in \mathbb{F}_q$ and the index pair $(j_1,j_2)$ is defined in the $j$-th row of the index array $\mathbf{P}_{i-r}$. The $r$ index arrays $\mathbf{P}_1,\ldots,\mathbf{P}_r$ are defined as follows:
	\begin{equation*}
	\hspace{-2.9cm}
	\mathbf{P_1}=
		\begin{bmatrix}
			(1, 1) & (1, 2) & \ldots & (1, k)\\
			(2, 1) & (2, 2) & \ldots & (2, k)\\
			\vdotswithin{1} & \vdotswithin{\alpha_n} & \ddots & \vdotswithin{{\alpha_n}^{k-1}}\\
			(\alpha, 1) & (\alpha, 2) & \ldots & (\alpha, k)\\
		\end{bmatrix},
	\end{equation*}
	\vspace{-0.25cm}
	$$\ \ \ \ \ \ \ \ \ \ \ \ \ \ \ \ \ \ \ \ \ \ \ \ \ \ \ \ \ \ \ \ \ \ \ \ \ \ \ \ \ \overbrace{\ \ \ \ \ \ \ \ \ \ \ \ \ \ \ \ \ \ \ \ }^{\lceil \frac{k}{r} \rceil}$$
	\begin{equation*}
	\mathbf{P}_i=
	\begin{bmatrix}
	(1, 1) & (1, 2) & \ldots & (1, k) &  (?, ?) & \ldots & (?, ?) \\
	(2, 1) & (2, 2) & \ldots & (2, k) & (?, ?) & \ldots & (?, ?) \\
	\vdotswithin{1} & \vdotswithin{\alpha_n} & \ddots & \vdotswithin{{\alpha_n}^{k-1}}\\
	(\alpha, 1) & (\alpha, 2) & \ldots & (\alpha, k) & (?, ?) & \ldots & (?, ?) \\
	\end{bmatrix}.
	\end{equation*}
	where the values of the indexes $(?, ?)$ are determined by a scheduling algorithm that guarantees the code is MDS, i.e. the entire information $\mathbf{x}$ can be recovered from any $k$ out of the $n$ vectors $\mathbf{c}_i$. $\blacksquare$
\end{definition}

One scheduling algorithm for Def. \ref{HashTagCodes} is defined in \cite{7463553,journals/corr/KralevskaGJO16}.

\begin{example}\label{Ex:HashTag0906alfa09}
	The linear expressions for the parity parts for a $(9, 6)$ HashTag code with $\alpha=9$ are given here. The way how we obtain them is explained in Section 4.1 in \cite{journals/corr/KralevskaGJO16}. We give one set of coefficients $f_{\nu,j, i}$ for equation (\ref{LinEquations}) from the finite field $\mathbf{F}_{32}$ with irreducible polynomial $x^5+x^3+1$. This code achieves the lower bound of repair bandwidth in (\ref{optimalMSR}), i.e. the repair bandwidth is $\gamma = \frac{8}{3} = 2.67$ for repair of any systematic node.
	
	\shorten{
	Due to the big size $54 \times 81$, the systematic generator matrix of the associated scalar code is presented graphically in Fig. \ref{G0906alpha09} instead of presenting it numerically.
	}

	{\small
		\begin{tabular}{llll}
			\renewcommand{\arraystretch}{0.9}
			\hspace{-1.0cm}\begin{tabular}{l@{}l@{}l@{}l@{}l@{}l@{}l@{}l@{}l@{}l@{}l@{}l}
				$c_{1,7}=$ & $\textbf{\ 7}x_{1,1}$ &+&$\textbf{10}x_{1,2}$ &+& $\textbf{18}x_{1,3}$ &+& $\textbf{11}x_{1,4}$ &+& $\textbf{17}x_{1,5}$ &+& $\textbf{\ 6}x_{1,6}$\\
				$c_{2,7}=$ & $\textbf{26}x_{2,1}$ &+& $\textbf{17}x_{2,2}$ &+& $\textbf{25}x_{2,3}$ &+& $\textbf{27}x_{2,4}$ &+& $\textbf{31}x_{2,5}$ &+& $\textbf{\ 4}x_{2,6}$\\
				$c_{3,7}=$ & $\textbf{22}x_{3,1}$ &+& $\textbf{12}x_{3,2}$ &+& $\textbf{27}x_{3,3}$ &+& $\textbf{31}x_{3,4}$ &+& $\textbf{31}x_{3,5}$ &+& $\textbf{23}x_{3,6}$\\
				$c_{4,7}=$ & $\textbf{17}x_{4,1}$ &+& $\textbf{\ 9}x_{4,2}$ &+& $\textbf{14}x_{4,3}$ &+& $\textbf{\ 4}x_{4,4}$ &+& $\textbf{21}x_{4,5}$ &+& $\textbf{25}x_{4,6}$\\
				$c_{5,7}=$ & $\textbf{20}x_{5,1}$ &+& $\textbf{\ 5}x_{5,2}$ &+& $\textbf{\ 5}x_{5,3}$ &+& $\textbf{13}x_{5,4}$ &+& $\textbf{11}x_{5,5}$ &+& $\textbf{16}x_{5,6}$\\
				$c_{6,7}=$ & $\textbf{25}x_{6,1}$ &+& $\textbf{16}x_{6,2}$ &+& $\textbf{30}x_{6,3}$ &+& $\textbf{28}x_{6,4}$ &+& $\textbf{10}x_{6,5}$ &+& $\textbf{24}x_{6,6}$\\
				$c_{7,7}=$ & $\textbf{20}x_{7,1}$ &+& $\textbf{\ 8}x_{7,2}$ &+& $\textbf{21}x_{7,3}$ &+& $\textbf{\ 9}x_{7,4}$ &+& $\textbf{\ 3}x_{7,5}$ &+& $\textbf{25}x_{7,6}$\\
				$c_{8,7}=$ & $\textbf{23}x_{8,1}$ &+& $\textbf{\ 4}x_{8,2}$ &+& $\textbf{12}x_{8,3}$ &+& $\textbf{16}x_{8,4}$ &+& $\textbf{\ 8}x_{8,5}$ &+& $\textbf{17}x_{8,6}$\\
				$c_{9,7}=$ & $\textbf{\ 2}x_{9,1}$ &+& $\textbf{21}x_{9,2}$ &+& $\textbf{\ 8}x_{9,3}$ &+& $\textbf{16}x_{9,4}$ &+& $\textbf{\ 7}x_{9,5}$ &+& $\textbf{25}x_{9,6}$\\
			\end{tabular} & \\
			\ 
			\\
			\renewcommand{\arraystretch}{0.9}
			\centering
			\hspace{-1.0cm}\begin{tabular}{l@{}l@{}l@{}l@{}l@{}l@{}l@{}l@{}l@{}l@{}l@{}l@{}l@{}l@{}l@{}l}
				$c_{1,8}=$ & $\textbf{\ 8}x_{1,1}$ &+&$\textbf{24}x_{1,2}$ &+& $\textbf{21}x_{1,3}$ &+& $\textbf{19}x_{1,4}$ &+& $\textbf{\ 6}x_{1,5}$ &+& $\textbf{20}x_{1,6}$ &+& $\textbf{\ 8} {x_{4,1}}$ &+& $\textbf{\ 6}{x_{2,4}}$\\		
				$c_{2,8}=$ & $\textbf{\ 3}x_{2,1}$ &+& $\textbf{12}x_{2,2}$ &+& $\textbf{\ 6}x_{2,3}$ &+& $\textbf{\ 3}x_{2,4}$ &+& $\textbf{16}x_{2,5}$ &+& $\textbf{10}x_{2,6}$ &+& $\textbf{30}{x_{5,1}}$ &+& $\textbf{24}{x_{1,5}}$\\
				$c_{3,8}=$ & $\textbf{23}x_{3,1}$ &+& $\textbf{20}x_{3,2}$ &+& $\textbf{30}x_{3,3}$ &+& $\textbf{\ 7}x_{3,4}$ &+& $\textbf{16}x_{3,5}$ &+& $\textbf{10}x_{3,6}$ &+& $\textbf{21}{x_{6,1}}$ &+& $\textbf{27}{x_{1,6}}$\\
				$c_{4,8}=$ & $\textbf{14}x_{4,1}$ &+& $\textbf{\ 7}x_{4,2}$ &+& $\textbf{10}x_{4,3}$ &+& $\textbf{14}x_{4,4}$ &+& $\textbf{24}x_{4,5}$ &+& $\textbf{20}x_{4,6}$ &+& $\textbf{16}{x_{1,2}}$ &+& $\textbf{31}{x_{5,4}}$\\
				$c_{5,8}=$ & $\textbf{25}x_{5,1}$ &+& $\textbf{11}x_{5,2}$ &+& $\textbf{29}x_{5,3}$ &+& $\textbf{12}x_{5,4}$ &+& $\textbf{20}x_{5,5}$ &+& $\textbf{24}x_{5,6}$ &+& $\textbf{15}{x_{2,2}}$ &+& $\textbf{\ 6}{x_{4,5}}$\\
				$c_{6,8}=$ & $\textbf{17}x_{6,1}$ &+& $\textbf{27}x_{6,2}$ &+& $\textbf{\ 4}x_{6,3}$ &+& $\textbf{21}x_{6,4}$ &+& $\textbf{15}x_{6,5}$ &+& $\textbf{11}x_{6,6}$ &+& $\textbf{19}{x_{3,2}}$ &+& $\textbf{21}{x_{4,6}}$\\
				$c_{7,8}=$ & $\textbf{19}x_{7,1}$ &+& $\textbf{23}x_{7,2}$ &+& $\textbf{16}x_{7,3}$ &+& $\textbf{\ 4}x_{7,4}$ &+& $\textbf{14}x_{7,5}$ &+& $\textbf{16}x_{7,6}$ &+& $\textbf{\ 9}{x_{1,3}}$ &+& $\textbf{\ 8}{x_{8,4}}$\\
				$c_{8,8}=$ & $\textbf{\ 5}x_{8,1}$ &+& $\textbf{26}x_{8,2}$ &+& $\textbf{22}x_{8,3}$ &+& $\textbf{30}x_{8,4}$ &+& $\textbf{22}x_{8,5}$ &+& $\textbf{21}x_{8,6}$ &+& $\textbf{24}{x_{2,3}}$ &+& $\textbf{26}{x_{7,5}}$\\
				$c_{9,8}=$ & $\textbf{10}x_{9,1}$ &+& $\textbf{\ 8}x_{9,2}$ &+& $\textbf{10}x_{9,3}$ &+& $\textbf{27}x_{9,4}$ &+& $\textbf{28}x_{9,5}$ &+& $\textbf{20}x_{9,6}$ &+& $\textbf{16}{x_{3,3}}$ &+& $\textbf{\ 4}{x_{7,6}}$\\
			\end{tabular} & \\
			\ 
			\\
			\renewcommand{\arraystretch}{0.9}
			\centering
			\hspace{-1.0cm}\begin{tabular}{l@{}l@{}l@{}l@{}l@{}l@{}l@{}l@{}l@{}l@{}l@{}l@{}l@{}l@{}l@{}l}
				$c_{1,9}=$ & $\textbf{20}x_{1,1}$ &+& $\textbf{20}x_{1,2}$ &+& $\textbf{30}x_{1,3}$ &+& $\textbf{17}x_{1,4}$ &+& $\textbf{12}x_{1,5}$ &+& $\textbf{27}x_{1,6}$ &+& $\textbf{28}{x_{7,1}}$ &+& $\textbf{\ 9}{x_{3,4}}$\\
				$c_{2,9}=$ & $\textbf{18}x_{2,1}$ &+& $\textbf{10}x_{2,2}$ &+& $\textbf{20}x_{2,3}$ &+& $\textbf{21}x_{2,4}$ &+& $\textbf{13}x_{2,5}$ &+& $\textbf{\ 7}x_{2,6}$ &+& $\textbf{\ 2}{x_{8,1}}$ &+& $\textbf{\ 6}{x_{3,5}}$\\
				$c_{3,9}=$ & $\textbf{31}x_{3,1}$ &+& $\textbf{25}x_{3,2}$ &+& $\textbf{12}x_{3,3}$ &+& $\textbf{18}x_{3,4}$ &+& $\textbf{15}x_{3,5}$ &+& $\textbf{24}x_{3,6}$ &+& $\textbf{31}{x_{9,1}}$ &+& $\textbf{28}{x_{2,6}}$\\				
				$c_{4,9}=$ & $\textbf{\ 6}x_{4,1}$ &+& $\textbf{16}x_{4,2}$ &+& $\textbf{26}x_{4,3}$ &+& $\textbf{\ 4}x_{4,4}$ &+& $\textbf{21}x_{4,5}$ &+& $\textbf{27}x_{4,6}$ &+& $\textbf{26}{x_{7,2}}$ &+& $\textbf{\ 8}{x_{6,4}}$\\
				$c_{5,9}=$ & $\textbf{\ 7}x_{5,1}$ &+& $\textbf{\ 6}x_{5,2}$ &+& $\textbf{26}x_{5,3}$ &+& $\textbf{\ 6}x_{5,4}$ &+& $\textbf{15}x_{5,5}$ &+& $\textbf{16}x_{5,6}$ &+& $\textbf{28}{x_{8,2}}$ &+& $\textbf{\ 4}{x_{6,5}}$\\
				$c_{6,9}=$ & $\textbf{20}x_{6,1}$ &+& $\textbf{20}x_{6,2}$ &+& $\textbf{12}x_{6,3}$ &+& $\textbf{20}x_{6,4}$ &+& $\textbf{18}x_{6,5}$ &+& $\textbf{26}x_{6,6}$ &+& $\textbf{19}{x_{9,2}}$ &+& $\textbf{30}{x_{5,6}}$\\
				$c_{7,9}=$ & $\textbf{26}x_{7,1}$ &+& $\textbf{\ 2}x_{7,2}$ &+& $\textbf{\ 6}x_{7,3}$ &+& $\textbf{20}x_{7,4}$ &+& $\textbf{17}x_{7,5}$ &+& $\textbf{23}x_{7,6}$ &+& $\textbf{\ 8} {x_{4,3}}$ &+& $\textbf{31}{x_{9,4}}$\\
				$c_{8,9}=$ & $\textbf{20}x_{8,1}$ &+& $\textbf{15}x_{8,2}$ &+& $\textbf{13}x_{8,3}$ &+& $\textbf{20}x_{8,4}$ &+& $\textbf{10}x_{8,5}$ &+& $\textbf{24}x_{8,6}$ &+& $\textbf{31}{x_{5,3}}$ &+& $\textbf{\ 9}{x_{9,5}}$\\
				$c_{9,9}=$ & $\textbf{\ 6}x_{9,1}$ &+& $\textbf{\ 2}x_{9,2}$ &+& $\textbf{31}x_{9,3}$ &+& $\textbf{12}x_{9,4}$ &+& $\textbf{16}x_{9,5}$ &+& $\textbf{30}x_{9,6}$ &+& $\textbf{20}{x_{6,3}}$ &+& $\textbf{13}{x_{8,6}}$\\
			\end{tabular} & \\ \ \\
		\end{tabular}
	}
\end{example}

\begin{figure}
	\begin{minipage}[b]{0.5\linewidth}
		\centering
		\includegraphics[width=8.5cm,height=4cm]{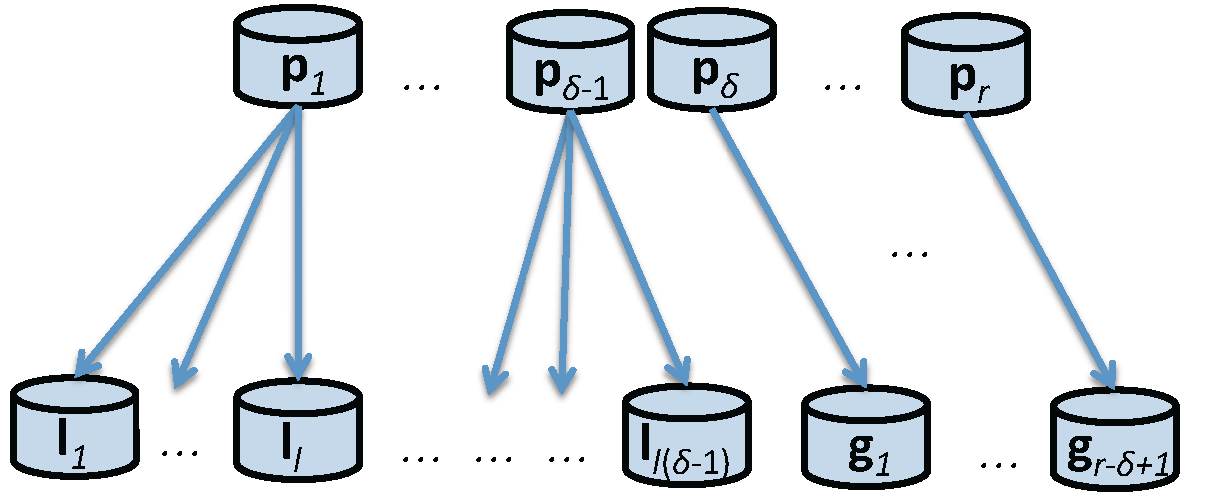}
	\end{minipage}
	\caption{There are $r$ parity nodes from a systematic $(n, k)$ MDS code with a sub-packetization level $\alpha$. The parity splitting technique generates from every parity node $\mathbf{p}_1,\ldots, \mathbf{p}_{\delta - 1}$ $l$ local parity nodes and renames the parity nodes $\mathbf{p}_{\delta},\ldots, \mathbf{p}_{r}$ as global parity nodes $\mathbf{g}_{1},\ldots, \mathbf{g}_{r-\delta+1}$.}
	\vspace{-0.3cm}
	\label{framework}
\end{figure}

We adapt the parity-splitting code construction for designing codes with local regeneration described in \cite{6846301} for the specifics of HashTag codes. The construction is described in Construction \ref{HashTagLRCConstruction}. For simplifying the description, we take some of the parameters to have specific relations, although it is possible to define a similar construction with general values of the parameters. Namely, we take that $r | k$ and $r | \alpha$. We also take that the parameters for the information locality $(l,\delta)$ are such that $l | k$ and $\delta \leq r$.
\begin{construction}\label{HashTagLRCConstruction}
	 An input to the construction is a $(n, k)$ HashTag MDS code with a sub-packetization level $\alpha$. This input comes with the associated linear parity equations (\ref{LinEquations}), i.e. with the associated systematic generator matrix $G$. The MDS code can be, but it does not necessarily have to be a MSR code. Another input is the information locality $(l,\delta)$ that the constructed code with local regeneration will have. The following steps are performed:
	 \begin{description}
	 	\item[Step 1.]~Split $k$ systematic nodes into $l$ disjunctive subsets $S_i, i\in [l]$, where every set has $\frac{k}{l}$ nodes. While this splitting can be arbitrary, take the canonical splitting where $S_1 = \{1,\ldots, \frac{k}{l}\}$, $S_2 = \{\frac{k}{l}+1,\ldots, \frac{2 k}{l}\}$, $\ldots$, $S_{l} = \{\frac{(l-1)k}{l}+1,\ldots, k\}$.  
	 	\item[Step 2.]~Split each of the $\alpha$ linear equations for the first $\delta - 1$ parity expressions (\ref{LinEquations}) into $l$ sub-summands where the variables in each equation correspond to the elements from the disjunctive subsets.
		\item[Step 3.]~Associate the obtained $\alpha \times l \times (\delta - 1)$ sub-summands to $l \times (\delta - 1)$ new local parity nodes.
		\item[Step 4.]~Rename the remaining $r - \delta + 1$ parity nodes that were not split in Step 1 - Step 3 as new global parity nodes.
		\item[Step 5.]~Obtain a new systematic generator matrix $G'$ from the local and global parity nodes.
		\item[Step 6.]~Return $G'$ as a generator matrix of a $[n,K=k\alpha,d_{\min},\alpha,k]$  vector code with information locality $(l,\delta)$.
	 \end{description}
		
\end{construction}

A graphical presentation of the parity-splitting procedure is given in Fig. \ref{framework}.

\begin{theorem}\label{Thm:LocalMSRHashTag}
	If the used $(n, k)$ MDS HashTag code in Construction \ref{HashTagLRCConstruction} is MSR, then the obtained $[n,K=k\alpha,d_{\min},\alpha,k]$ code with information locality $(l,\delta)$ is a MSR-Local code, where  
	\begin{equation}
	d_{\min} = n-k+1-\left( \frac{k}{l} - 1 \right)(\delta - 1) .
	\label{MinDistanceHashTagLocal}
	\end{equation}
\end{theorem}
\begin{proof}\it{(Sketch)} 
	{\rm Since in Construction \ref{HashTagLRCConstruction} we took that $r | k$ and $r | \alpha$, it means that the scalar dimension of the code is $K=m l \alpha$ for some integer $m$. Then the proof continues basically as a technical adaptation of the proof of Theorem 5.5 that Kamath et al. gave for the pyramid-like MSR-Local codes constructed with the parity-splitting strategy in \cite{6846301}.}
\end{proof}
Note that if $\alpha < r^{\frac{k}{r}}$, then HashTag codes are sub-optimal in terms of the repair bandwidth. Consequently, the produced codes with Construction 1 are locally repairable, but they are not MSR-Local codes.
\begin{example}\label{Ex:HashTag0906alfa09-Split}
Let us split the MSR code given in Example \ref{Ex:HashTag0906alfa09} into a code with local regeneration and with information locality $(l=2,\delta=2)$. In Step 1 we split 6 systematic nodes $\{\mathbf{c}_{1},\ldots, \mathbf{c}_{6} \}$ into $l=2$ disjunctive subsets $S_1 = \{\mathbf{c}_{1}, \mathbf{c}_{2}, \mathbf{c}_{3} \}$ and $S_2 = \{\mathbf{c}_{4}, \mathbf{c}_{5}, \mathbf{c}_{6} \}$. According to Step 2 of Construction 1, the first global parity $\mathbf{c}_{7}$ in Example \ref{Ex:HashTag0906alfa09} is split into two local parities $\mathbf{l}_{1}=(l_{1,1}, \ldots,l_{9,1})^T$ and $\mathbf{l}_{2}=(l_{1,2},\ldots,l_{9,2})^T$ as follows:
\\
\\
	{\small
		\begin{tabular}{llll}
			\hspace{-0.5cm}
			\begin{tabular}{l@{}l@{}l@{}l@{}l@{}l@{}l@{}l@{}l@{}l@{}l@{}l}
				$l_{1,1}=$ & $\textbf{\ 7}x_{1,1}$ &+&$\textbf{10}x_{1,2}$ &+& $\textbf{18}x_{1,3}$\\
				$l_{2,1}=$ & $\textbf{26}x_{2,1}$ &+& $\textbf{17}x_{2,2}$ &+& $\textbf{25}x_{2,3}$\\
				$l_{3,1}=$ & $\textbf{22}x_{3,1}$ &+& $\textbf{12}x_{3,2}$ &+& $\textbf{27}x_{3,3}$\\
				$l_{4,1}=$ & $\textbf{17}x_{4,1}$ &+& $\textbf{\ 9}x_{4,2}$ &+& $\textbf{14}x_{4,3}$\\
				$l_{5,1}=$ & $\textbf{20}x_{5,1}$ &+& $\textbf{\ 5}x_{5,2}$ &+& $\textbf{\ 5}x_{5,3}$\\
				$l_{6,1}=$ & $\textbf{25}x_{6,1}$ &+& $\textbf{16}x_{6,2}$ &+& $\textbf{30}x_{6,3}$\\
				$l_{7,1}=$ & $\textbf{20}x_{7,1}$ &+& $\textbf{\ 8}x_{7,2}$ &+& $\textbf{21}x_{7,3}$\\
				$l_{8,1}=$ & $\textbf{23}x_{8,1}$ &+& $\textbf{\ 4}x_{8,2}$ &+& $\textbf{12}x_{8,3}$\\
				$l_{9,1}=$ & $\textbf{\ 2}x_{9,1}$ &+& $\textbf{21}x_{9,2}$ &+& $\textbf{\ 8}x_{9,3}$\\
			\end{tabular} & \ \ & 
			\begin{tabular}{l@{}l@{}l@{}l@{}l@{}l@{}l@{}l@{}l@{}l@{}l@{}l}
				$l_{1,2}=$ & $\textbf{11}x_{1,4}$ &+& $\textbf{17}x_{1,5}$ &+& $\textbf{\ 6}x_{1,6}$\\
				$l_{2,2}=$ & $\textbf{27}x_{2,4}$ &+& $\textbf{31}x_{2,5}$ &+& $\textbf{\ 4}x_{2,6}$\\
				$l_{3,2}=$ & $\textbf{31}x_{3,4}$ &+& $\textbf{31}x_{3,5}$ &+& $\textbf{23}x_{3,6}$\\
				$l_{4,2}=$ & $\textbf{\ 4}x_{4,4}$ &+& $\textbf{21}x_{4,5}$ &+& $\textbf{25}x_{4,6}$\\
				$l_{5,2}=$ & $\textbf{13}x_{5,4}$ &+& $\textbf{11}x_{5,5}$ &+& $\textbf{16}x_{5,6}$\\
				$l_{6,2}=$ & $\textbf{28}x_{6,4}$ &+& $\textbf{10}x_{6,5}$ &+& $\textbf{24}x_{6,6}$\\
				$l_{7,2}=$ & $\textbf{\ 9}x_{7,4}$ &+& $\textbf{\ 3}x_{7,5}$ &+& $\textbf{25}x_{7,6}$\\
				$l_{8,2}=$ & $\textbf{16}x_{8,4}$ &+& $\textbf{\ 8}x_{8,5}$ &+& $\textbf{17}x_{8,6}$\\
				$l_{9,2}=$ & $\textbf{16}x_{9,4}$ &+& $\textbf{\ 7}x_{9,5}$ &+& $\textbf{25}x_{9,6}$\\
			\end{tabular}
		\end{tabular}		
	}
\\
	The remaining two global parities are kept as they are given in Example \ref{Ex:HashTag0906alfa09}, they are only renamed as $\mathbf{g}_{1}=(c_{1,8}, c_{2,8},\ldots,c_{9,8})^T$ and $\mathbf{g}_{2}=(c_{1,9}, c_{2,9},\ldots,c_{9,9})^T$. The overall code is a $(10, 6)$ code or with the terminology from \cite{conf/usenix/HuangSXOCG0Y12} it is a $(6, 2, 2)$ code. $\blacksquare$
\vspace{0.2cm}	
\end{example}

\begin{example}\label{Ex:HashTag0906alfa09-Split3}
	Let us split the same MSR code now with parameters $(l=3,\delta=2)$. In Step 1 we split 6 systematic nodes $\{\mathbf{c}_{1},\ldots, \mathbf{c}_{6} \}$ into $l=3$ disjunctive subsets $S_1 = \{\mathbf{c}_{1}, \mathbf{c}_{2}\}$, $S_2 = \{\mathbf{c}_{3}, \mathbf{c}_{4}\}$ and $S_3 = \{\mathbf{c}_{5}, \mathbf{c}_{6}\}$. In Step 2 of Construction 1, the first global parity $\mathbf{c}_{7}$ is split into three local parities: $\mathbf{l}_{1}=(l_{1,1}, \ldots,l_{9,1})^T$,  $\mathbf{l}_{2}=(l_{1,2},\ldots,l_{9,2})^T$ and 
	$\mathbf{l}_{3}=(l_{1,3},\ldots,l_{9,3})^T$	as follows:
	\\
	\\
	{\small
	\begin{tabular}{l@{\ }l@{}l@{\ }l@{}l}
		\hspace{-0.5cm}
		\begin{tabular}{l@{}l@{}l@{}l@{}l@{}l@{}l@{}l@{}l@{}l@{}l@{}l}
			$l_{1,1}=$ & $\textbf{\ 7}x_{1,1}$ &+&$\textbf{10}x_{1,2}$    \\
			$l_{2,1}=$ & $\textbf{26}x_{2,1}$  &+& $\textbf{17}x_{2,2}$   \\
			$l_{3,1}=$ & $\textbf{22}x_{3,1}$  &+& $\textbf{12}x_{3,2}$   \\
			$l_{4,1}=$ & $\textbf{17}x_{4,1}$  &+& $\textbf{\ 9}x_{4,2}$  \\
			$l_{5,1}=$ & $\textbf{20}x_{5,1}$  &+& $\textbf{\ 5}x_{5,2}$  \\
			$l_{6,1}=$ & $\textbf{25}x_{6,1}$  &+& $\textbf{16}x_{6,2}$   \\
			$l_{7,1}=$ & $\textbf{20}x_{7,1}$  &+& $\textbf{\ 8}x_{7,2}$  \\
			$l_{8,1}=$ & $\textbf{23}x_{8,1}$  &+& $\textbf{\ 4}x_{8,2}$  \\
			$l_{9,1}=$ & $\textbf{\ 2}x_{9,1}$ &+& $\textbf{21}x_{9,2}$   \\
		\end{tabular} && 
		\begin{tabular}{l@{}l@{}l@{}l@{}l@{}l@{}l@{}l@{}l@{}l@{}l@{}l}
			$l_{1,2}=$ & $\textbf{18}x_{1,3}$  &+& $\textbf{11}x_{1,4}$ \\
			$l_{2,2}=$ & $\textbf{25}x_{2,3}$  &+& $\textbf{27}x_{2,4}$ \\
			$l_{3,2}=$ & $\textbf{27}x_{3,3}$  &+& $\textbf{31}x_{3,4}$ \\
			$l_{4,2}=$ & $\textbf{14}x_{4,3}$  &+& $\textbf{\ 4}x_{4,4}$\\
			$l_{5,2}=$ & $\textbf{\ 5}x_{5,3}$ &+& $\textbf{13}x_{5,4}$ \\
			$l_{6,2}=$ & $\textbf{30}x_{6,3}$  &+& $\textbf{28}x_{6,4}$ \\
			$l_{7,2}=$ & $\textbf{21}x_{7,3}$  &+& $\textbf{\ 9}x_{7,4}$\\
			$l_{8,2}=$ & $\textbf{12}x_{8,3}$  &+& $\textbf{16}x_{8,4}$ \\
			$l_{9,2}=$ & $\textbf{\ 8}x_{9,3}$ &+& $\textbf{16}x_{9,4}$ \\
		\end{tabular} &&
		\begin{tabular}{l@{}l@{}l@{}l@{}l@{}l@{}l@{}l@{}l@{}l@{}l@{}l}
			$l_{1,3}=$ & $\textbf{17}x_{1,5}$  &+& $\textbf{\ 6}x_{1,6}$\\
			$l_{2,3}=$ & $\textbf{31}x_{2,5}$  &+& $\textbf{\ 4}x_{2,6}$\\
			$l_{3,3}=$ & $\textbf{31}x_{3,5}$  &+& $\textbf{23}x_{3,6}$\\
			$l_{4,3}=$ & $\textbf{21}x_{4,5}$  &+& $\textbf{25}x_{4,6}$\\
			$l_{5,3}=$ & $\textbf{11}x_{5,5}$  &+& $\textbf{16}x_{5,6}$\\
			$l_{6,3}=$ & $\textbf{10}x_{6,5}$  &+& $\textbf{24}x_{6,6}$\\
			$l_{7,3}=$ & $\textbf{\ 3}x_{7,5}$ &+& $\textbf{25}x_{7,6}$\\
			$l_{8,3}=$ & $\textbf{\ 8}x_{8,5}$ &+& $\textbf{17}x_{8,6}$\\
			$l_{9,3}=$ & $\textbf{\ 7}x_{9,5}$ &+& $\textbf{25}x_{9,6}$\\
		\end{tabular}
	\end{tabular}		
	}
\\
	The remaining two global parities are kept as they are given in Example \ref{Ex:HashTag0906alfa09}, but they are just renamed as $\mathbf{g}_{1}=(c_{1,8}, c_{2,8},\ldots,c_{9,8})^T$ and $\mathbf{g}_{2}=(c_{1,9}, c_{2,9},\ldots,c_{9,9})^T$. The overall code is a $(11, 6)$ code or with the terminology from \cite{conf/usenix/HuangSXOCG0Y12} it is a $(6, 3, 2)$ code. $\blacksquare$
	\vspace{0.2cm}	
\end{example}

\vspace{-0.2cm}
There are two interesting aspects of Theorem \ref{Thm:LocalMSRHashTag} that should be emphasized: {\bf 1.} We give an explicit construction of an MSR-Local code (note that in \cite{6846301} the construction is existential), and {\bf 2.} Examples \ref{Ex:HashTag0906alfa09-Split} and \ref{Ex:HashTag0906alfa09-Split3} show that the size of the finite field can be slightly lower than the size proposed in \cite{6846301}. Namely, the MSR HashTag code used in our example is defined over $\mathbf{F}_{32}$, while the lower bound in \cite{6846301} suggests the field size to be bigger than $\binom{9}{6} = 84$. We consider this as a minor contribution and an indication that a deeper theoretical analysis can further lower the field size bound given in \cite{6846301}.

\vspace{-0.05cm}
\section{Repair Duality}\label{repair}
\begin{theorem}\label{Thm:DoubleNature}
		Let $\mathcal{C}$ be a $(n, k)$ MSR HashTag code with $\gamma_{MSR}^{min}=\frac{M}{k}\frac{n-1}{n-k}$. Further, let $\mathcal{C}'$ be a $[n,K=k\alpha,d_{\min},\alpha,k]$ code with local regeneration and with information locality $(l,\delta)$ obtained by Construction \ref{HashTagLRCConstruction}. If we denote with $\gamma_{Local}^{min}$ the minimum repair bandwidth for single systematic node repair with $\mathcal{C}'$, then  
		\vspace{-0.25cm}
		\begin{equation}
		\gamma_{Local}^{min}=\min(\frac{M}{k}\frac{\frac{k}{l}+\delta - 2}{\delta - 1}, \frac{M}{k}\frac{n-1}{n-k}).
		\label{optimalBW}
		\end{equation}
\end{theorem}
\vspace{-0.2cm}
\begin{proof}
	When repairing one systematic node, we can always treat local nodes as virtual global nodes from which they have been constructed by splitting. Then with the use of other global nodes we have a situation of repairing one systematic node in the original MSR code for which the repair bandwidth is $\frac{M}{k}\frac{n-1}{n-k}$. On the other hand, if we use the MSR-Local code, then we have the following situation. There are $\frac{k}{l}$ systematic nodes in the MSR-Local code, and the total length of the MSR-Local code is $\frac{k}{l} + \delta - 1$. The file size for the MSR-Local code is decreased by a factor $l$, i.e. it is $\frac{M}{l}$. If we apply the MSR repair bandwidth for these values we get:
	$$\frac{\frac{M}{l}}{\frac{k}{l}} \cdot \frac{\frac{k}{l}+(\delta - 1)-1}{\delta - 1} =  \frac{M}{k} \frac{\frac{k}{l}+\delta - 2}{\delta - 1}.$$
\end{proof}
		

Theorem \ref{Thm:DoubleNature} is one of the main contributions of this work: It emphasizes the \emph{repair duality} for repairing one systematic node: by the local and global parity nodes or only by the local parity nodes. We want to emphasize the practical importance of Theorem \ref{Thm:DoubleNature}. Namely, in practical implementations regardless of the theoretical value of $\gamma_{Local}^{min}$, the number of I/O operations and the access time for the contacted parts can be either crucial or insignificant. In those cases an intelligent repair strategy implemented in the distributed storage system can decide which repair procedure should be used: the one with global parity nodes or the one with the local parity nodes. We illustrate this by the following example.

\begin{example}\label{Ex:Duality}
	Let us consider the $(9, 6)$ MSR HashTag code given in Example \ref{Ex:HashTag0906alfa09} and its corresponding local variant from Construction \ref{HashTagLRCConstruction} with information locality $(l=2,\delta=2)$ given in Example \ref{Ex:HashTag0906alfa09-Split}. That means that the code with local regeneration has 6 systematic nodes, 2 local and 2 global parity nodes. 
	
	Let us analyze the number of reads when we recover one unavailable systematic node. If we recover with the local nodes, then we have to perform 3 sequential reads, reading the whole data in a contiguous manner from 3 nodes. If we repair the unavailable data with the help of both local and global parity nodes, it reduces to the case of recovery with a MSR code, where the number of sequential reads is between 8 and 24 (average 16 reads) but the amount of transferred data is equivalent to 2.67 nodes. 
	
	More concretely, let us assume that we want to recover the node $\mathbf{x}_1 = (x_{1,1}, x_{2,1},\ldots,x_{9,1})^T$. 
	\begin{enumerate}
		\item For a recovery only with the local parity $\mathbf{l}_{1}$, 3 sequential reads of $\mathbf{l}_{1}$, $\mathbf{x}_{2}$ and $\mathbf{x}_{3}$ are performed.
		\item For a recovery with the local and global parities:
		\begin{enumerate}
			\item First, read $l_{1,1}$,  $l_{2,1}$ and $l_{3,1}$ from $\mathbf{l}_{1}$, and $x_{1,2}$,  $x_{2,2}$ and $x_{3,2}$ from $\mathbf{x}_{2}$ and $x_{1,3}$, $x_{2,3}$ and $x_{3,3}$ from $\mathbf{x}_{3}$ to recover $x_{1,1}$,  $x_{2,1}$ and $x_{3,1}$.
			\item Additionally, read $x_{1,4}$,  $x_{2,4}$ and $x_{3,4}$ from $\mathbf{x}_{4}$ and $x_{1,5}$,  $x_{2,5}$ and $x_{3,5}$ from $\mathbf{x}_{5}$ and $x_{1,6}$, $x_{2,6}$ and $x_{3,6}$ from $\mathbf{x}_{6}$.
			\item Then, read $c_{1,8}$,  $c_{2,8}$ and $c_{3,8}$ from the global parity $\mathbf{g}_{1}$ to recover $x_{4,1}$,  $x_{5,1}$ and $x_{6,1}$.
			\item Finally, read $c_{1,9}$,  $c_{2,9}$ and $c_{3,9}$ from the global parity $\mathbf{g}_{2}$ to recover $x_{7,1}$,  $x_{8,1}$ and $x_{9,1}$.
		\end{enumerate}
	\end{enumerate}

	Now, let a small file of 54 KB be stored across 6 systematic, 2 local and 2 global parity nodes. The sub-packetization level is $\alpha=9$, thus every node stores 9 KB, sub-packetized in 9 parts, each of size 1 KB. If the access time for starting a read operation is approximately the same as transferring 9 KB, then repairing with local and global parity nodes is more expensive since we have to perform in average 12 reads, although the amount of transferred data is equivalent to 2.67 nodes.
	
	On the other hand, let us have a big file of 540 MB stored across 6 systematic nodes and 2 local and 2 global parity nodes. The sub-packetization level is again $\alpha=9$, thus every node stores 90 MB, sub-packetized in 9 parts, each of size 10 MB. The access time for starting a read operation is again approximately the same as transferring 9 KB, which is insignificant in comparison with the total amount of transferred data in the process of repairing of a node. In this case, it is better to repair a failed node with local and global parity nodes since it requires a transfer of 240 MB versus the repair just with local nodes that requires a transfer of 270 MB.
\end{example}

\section{Experiments in Hadoop}\label{Hadoop}
The repair duality discussed in Example \ref{Ex:Duality} of previous section was mainly influenced by one system characteristic: the access time for starting a read operation. In different environments of distributed storage systems there are several similar system characteristics that can affect the repair duality and its final optimal procedure. We next discuss this matter for Hadoop.

Hadoop is an open-source software framework used for distributed storage and processing of big data sets \cite{white2012hadoop}. From release 3.0.0-alpha2 Hadoop offers several erasure codes such as $(9, 6)$ and $(14, 10)$ Reed-Solomon (RS) codes. Hadoop Distributed File System (HDFS) has the concepts of \emph{Splits} and \emph{Blocks}. A Split is a logical representation of the data while a Block describes the physical alignment of data. Splits and Blocks in Hadoop are user defined: a logical split can be composed of multiple blocks and one block can have multiple splits. All these choices determine in a more complex way the access time for I/O operations.

To verify the performance of HashTag codes and their locally repairable and locally regenerating variants we implemented them in C/C++ and used them in HDFS. 

For the code $(9, 6)$ we used one NameNode, nine DataNodes, and one client node. All nodes had a size of 50 GB and were connected with a local network of 10 Gbps. The nodes were running on Linux machines equipped with Intel Xeon E5-2676 v3 running on 2.4 GHz. We have experimented with different block sizes (90 MB and 360 MB), different split sizes (512 KB, 1 MB and 4 MB) and different sub-packetization levels ($\alpha=1, 3, 6,$ and $9$) in order to check how they affect the repair time of one lost node. The measured times to recover one node are presented in Fig. \ref{9_6_Hadoop01}. Note that the sub-packetization level $\alpha=1$ represents the RS code that is available in HDFS, while for every other $\alpha = 3, 6, 9$ the codes are HashTag codes. In all measurements HashTag codes outperform RS. The cost of having significant number of I/O operations for the sub-packetization level $\alpha=9$ is the highest for the smallest block and split size (Block size of 90 MB and Split size of 512 KB). This is shown by yellow bars. As the split sizes increase, the disadvantage of bigger number of I/Os due to the increased sub-packetization diminishes, and the repair time decreases further (red and blue bars). 
\begin{figure}[h!]
	\includegraphics[width=3.6in]{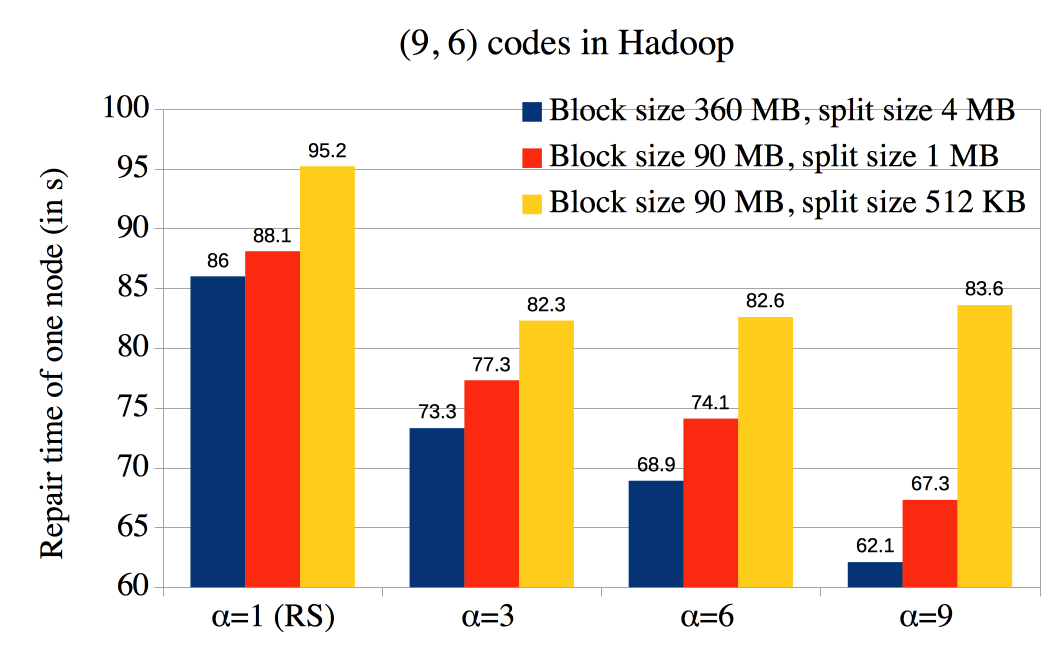}
	\vspace{-0.8cm}
	\caption{Experiments in HDFS. Time to repair one lost node of 50 GB with a $(9, 6)$ code for different sub-packetization levels $\alpha$. Note that the RS code for $\alpha=1$ is available in the latest release 3.0.0-alpha2 of Apache Hadoop.}
	\label{9_6_Hadoop01}
\end{figure}

In Fig. \ref{9_6_HadoopAlpha9}, we compare the repair times for one lost node of 50 GB with the codes from Examples \ref{Ex:HashTag0906alfa09-Split} and \ref{Ex:HashTag0906alfa09-Split3}. The cost of bigger redundancy with the locally repairable code $(10,6)$ (which is also a locally regenerative code since the sub-packetization level is $\alpha=9$) with $(l=2,\delta=2)$ (the red bar) is still not enough to outperform the ordinary $(9, 6)$ HashTag MSR code (the blue bar). However, paying even higher cost by increasing the redundancy for the other locally regenerative code $(11,6)$ with $(l=3,\delta=2)$ (the yellow bar) finally manages to outperform the repairing time for the ordinary $(9, 6)$ HashTag MSR code.
\begin{figure}[h!]
	\includegraphics[width=3.6in]{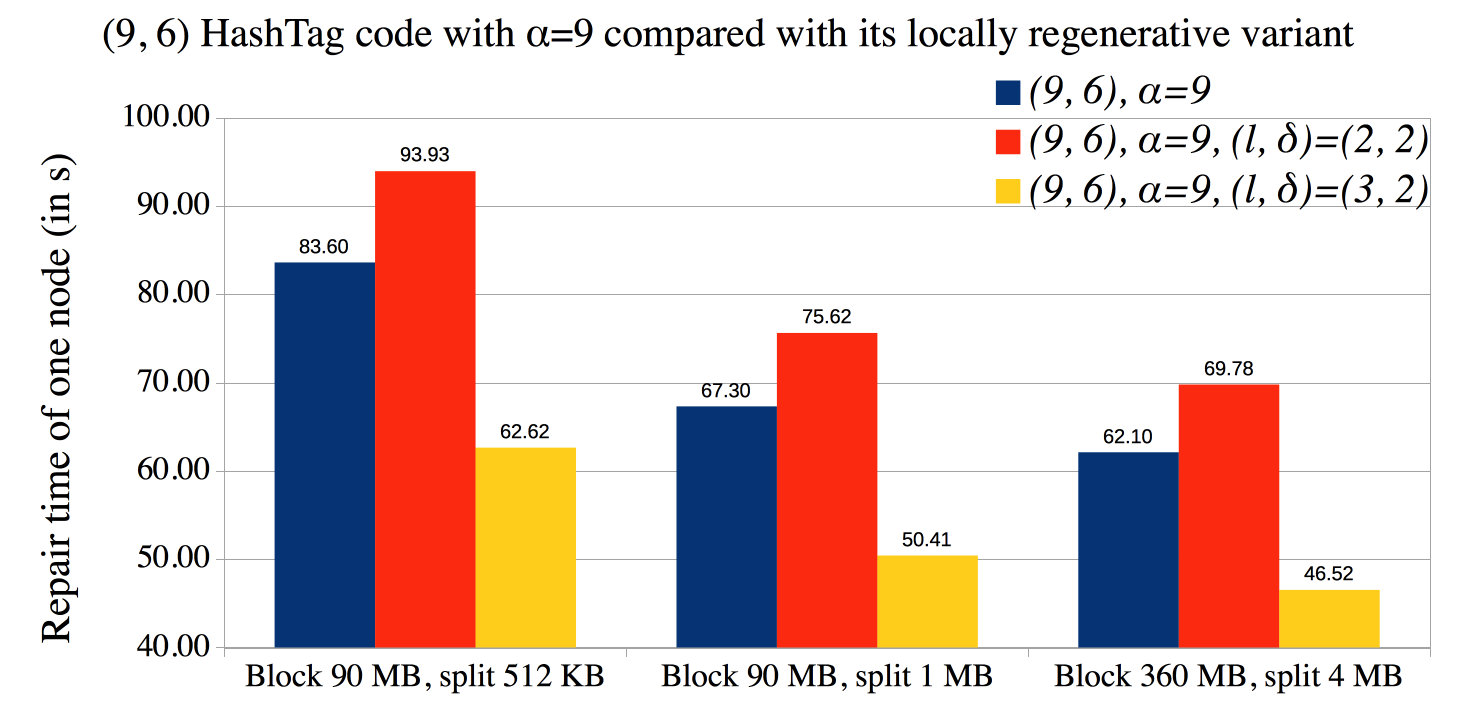}
	\vspace{-0.8cm}
	\caption{Comparison of repair times for one lost node of 50 GB for an ordinary $(9, 6)$ HashTag MSR code ($\alpha=9$), and its locally regenerative variants with $(l=2,\delta=2)$ and $(l=3,\delta=2)$.}
	\label{9_6_HadoopAlpha9}
\end{figure}

The situation of comparing the slightly less optimal $(9, 6)$ HashTag code with $\alpha=6$ with its locally repairable variants with $(l=2,\delta=2)$ and $(l=3,\delta=2)$ is different, and  this is presented in Fig. \ref{9_6_HadoopAlpha6}. In this case, all variants of locally repairable codes outperform the original HashTag code, i.e. they repair a failed node in shorter time than the original HashTag code from which they were constructed. 
\begin{figure}[h!]
	\includegraphics[width=3.6in]{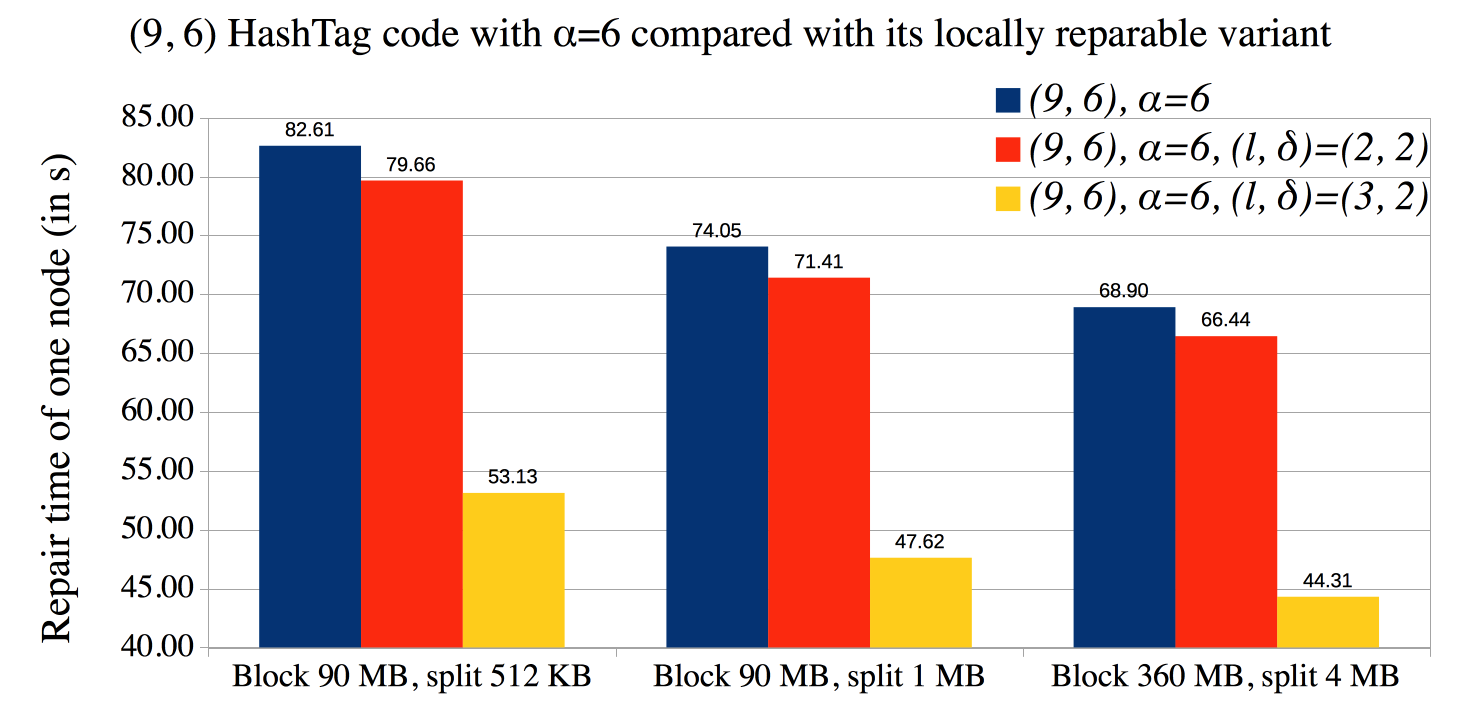}
	\vspace{-0.8cm}
	\caption{Comparison of repair times for one lost node of 50 GB for a $(9, 6)$ HashTag code with $\alpha=6$, and its locally repairable variants with $(l=2,\delta=2)$ and $(l=3,\delta=2)$.}
	\label{9_6_HadoopAlpha6}
\end{figure}

\section{Conclusions}\label{summary}
We constructed an explicit family of locally repairable and locally regenerating codes. We applied the technique of parity-splitting on HashTag codes and constructed codes with locality. For these codes we showed that there are two ways to repair a node (repair duality), and in practice which way is applied depends on optimization metrics such as the repair bandwidth, the number of I/O operations, the access time for the contacted parts and the size of the stored file. Additionally, we showed that the size of the finite field can be slightly lower than the theoretically obtained lower bound on the size in the literature. We hope that this work will inspire a further theoretical analysis for obtaining new lower bound.

\bibliography{refer}

\begin{thebibliography}{10}
\providecommand{\url}[1]{#1}
\csname url@samestyle\endcsname
\providecommand{\newblock}{\relax}
\providecommand{\bibinfo}[2]{#2}
\providecommand{\BIBentrySTDinterwordspacing}{\spaceskip=0pt\relax}
\providecommand{\BIBentryALTinterwordstretchfactor}{4}
\providecommand{\BIBentryALTinterwordspacing}{\spaceskip=\fontdimen2\font plus
\BIBentryALTinterwordstretchfactor\fontdimen3\font minus
  \fontdimen4\font\relax}
\providecommand{\BIBforeignlanguage}[2]{{%
\expandafter\ifx\csname l@#1\endcsname\relax
\typeout{** WARNING: IEEEtran.bst: No hyphenation pattern has been}%
\typeout{** loaded for the language `#1'. Using the pattern for}%
\typeout{** the default language instead.}%
\else
\language=\csname l@#1\endcsname
\fi
#2}}
\providecommand{\BIBdecl}{\relax}
\BIBdecl

\bibitem{5550492}
A.~G. Dimakis, P.~B. Godfrey, Y.~Wu, M.~J. Wainwright, and K.~Ramchandran,
  ``Network coding for distributed storage systems,'' IEEE Trans. Inf. Theory,
  vol.~56, no.~9, Sept. 2010, pp. 4539--4551.

\bibitem{journals/tit/GopalanHSY12}
P.~Gopalan, C.~Huang, H.~Simitci, and S.~Yekhanin, ``On the locality of
  codeword symbols,'' IEEE Trans. on Inf. Theory, vol.~58, no.~11, 2012, pp.
  6925--6934.

\bibitem{conf/infocom/OggierD11}
F.~E. Oggier and A.~Datta, ``Self-repairing homomorphic codes for distributed
  storage systems,'' in INFOCOM, 2011, pp. 1215--1223.

\bibitem{6195703}
D.~Papailiopoulos, J.~Luo, A.~Dimakis, C.~Huang, and J.~Li, ``Simple
  regenerating codes: Network coding for cloud storage,'' in IEEE INFOCOM,
  2012, pp. 2801--2805.

\bibitem{7463553}
K.~Kralevska, D.~Gligoroski, and H.~{\O}verby, ``General sub-packetized
  access-optimal regenerating codes,'' IEEE Communications Letters, vol.~20,
  no.~7, July 2016, pp. 1281--1284.

\bibitem{journals/corr/KralevskaGJO16}
K.~Kralevska, D.~Gligoroski, R.~E. Jensen, and H.~{\O}verby, ``Hashtag erasure
  codes: From theory to practice,'' arXiv preprint, vol. abs/1609.02450, 2016.

\bibitem{Huang:2013:PCF:2435204.2435207}
C.~Huang, M.~Chen, and J.~Li, ``Pyramid codes: Flexible schemes to trade space
  for access efficiency in reliable data storage systems,'' Trans. Storage,
  vol.~9, no.~1, Mar. 2013, pp. 3:1--3:28.

\bibitem{conf/usenix/HuangSXOCG0Y12}
C.~Huang, H.~Simitci, Y.~Xu, A.~Ogus, B.~Calder, P.~Gopalan, J.~Li, and
  S.~Yekhanin, ``Erasure coding in windows azure storage,'' in USENIX Annual
  Technical Conference, 2012, pp. 15--26.

\bibitem{journals/pvldb/SathiamoorthyAPDVCB13}
M.~Sathiamoorthy, M.~Asteris, D.~S. Papailiopoulos, A.~G. Dimakis, R.~Vadali,
  S.~Chen, and D.~Borthakur, ``{XOR}ing elephants: Novel erasure codes for big
  data,'' Proc. VLDB Endow., vol.~6, no.~5, 2013, pp. 325--336.

\bibitem{7593121}
K.~Kralevska, D.~Gligoroski, and H.~{\O}verby, ``Balanced locally repairable
  codes,'' in Int. Sym. on Turbo Codes and Iterative Inf. Processing (ISTC),
  Sept 2016, pp. 280--284.

\bibitem{6655894}
A.~S. Rawat, O.~O. Koyluoglu, N.~Silberstein, and S.~Vishwanath, ``Optimal
  locally repairable and secure codes for distributed storage systems,'' IEEE
  Trans. on Inf. Theory, vol.~60, no.~1, 2014, pp. 212--236.

\bibitem{6846301}
G.~M. Kamath, N.~Prakash, V.~Lalitha, and P.~V. Kumar, ``Codes with local
  regeneration and erasure correction,'' IEEE Trans. on Inf. Theory, vol.~60,
  no.~8, Aug 2014, pp. 4637--4660.

\bibitem{7282575}
I.~Ahmad and C.~C. Wang, ``When locally repairable codes meet regenerating
  codes - what if some helpers are unavailable,'' in IEEE Int. Symp. on Inf.
  Theory (ISIT), June 2015, pp. 849--853.

\bibitem{7541379}
T.~Ernvall, T.~Westerbäck, R.~Freij-Hollanti, and C.~Hollanti, ``A connection
  between locally repairable codes and exact regenerating codes,'' in IEEE Int.
  Symp. on Inf. Theory (ISIT), July 2016, pp. 650--654.

\bibitem{188446}
M.~Xia, M.~Saxena, M.~Blaum, and D.~A. Pease, ``A tale of two erasure codes in
  {HDFS},'' in 13th {USENIX} Conference on File and Storage Technologies
  ({FAST}), 2015, pp. 213--226.

\bibitem{white2012hadoop}
T.~White, Hadoop: The definitive guide.\hskip 1em plus 0.5em minus 0.4em\relax
  O'Reilly Media, Inc., 2012.

\end{thebibliography}
\bibliographystyle{IEEEtran}

\end{document}